\documentclass[final,12pt,3p]{elsarticle}

\usepackage{graphicx}
\usepackage{verbatim} 
\usepackage{amssymb}
\usepackage{amsthm}
\usepackage{amsmath}
\usepackage{algorithm}
\usepackage{algorithmic}
\usepackage{tikz}
\usepackage{lineno}
\usepackage{hyperref}
\usepackage{caption}
\usetikzlibrary{shapes}
\usepackage{subcaption}
\usepackage{geometry}
\usepackage{mathrsfs} 
\usepackage{array}
\usepackage{multirow}
\usepackage{booktabs} 
\usepackage{adjustbox}
\biboptions{sort&compress}
\hypersetup{
colorlinks=true,
linkcolor=red,
filecolor=green,
urlcolor=green,
anchorcolor=green,
citecolor=cyan,
pdftitle={Overleaf Example},
pdfpagemode=FullScreen,
}
\newtheorem{theorem}{Theorem}

\newtheorem{remark}{Remark}

\allowdisplaybreaks

\newcommand{\dt}[1]{\frac{\text{d}#1}{\text{d}t}}
 \newcommand{\pp}[2]{\frac{\partial #1}{\partial #2}} 
 \newcommand{\ppp}[2]{\frac{\partial^2 #1}{\partial {#2}^2}} 
\DeclareMathOperator{\gO}{O}

 \usepackage{xcolor}

\begin{document}

\begin{tikzpicture}[remember picture,overlay]
	\node[anchor=north east,inner sep=20pt] at (current page.north east)
	{\includegraphics[scale=0.2]{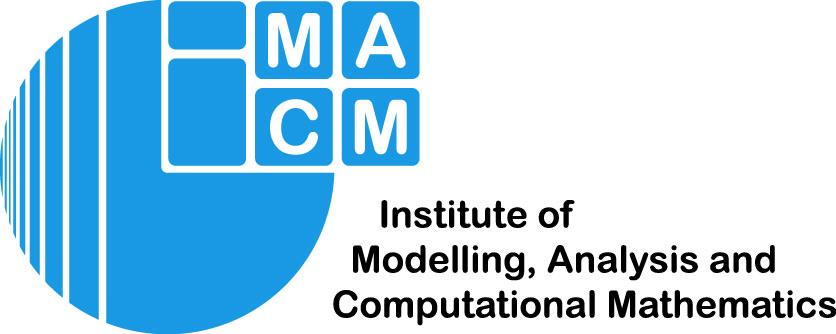}};
\end{tikzpicture}

\begin{frontmatter}

\title{A Nonstandard Finite Difference Scheme\\ for an SEIQR Epidemiological PDE Model}

\author[BUW,MAIS]{Achraf Zinihi}
\ead{a.zinihi@edu.umi.ac.ma} 

\author[BUW]{Matthias Ehrhardt\corref{Corr}}
\cortext[Corr]{Corresponding author}
\ead{ehrhardt@uni-wuppertal.de}

\author[MAIS]{Moulay Rchid Sidi Ammi}
\ead{rachidsidiammi@yahoo.fr}

\address[BUW]{University of Wuppertal, Applied and Computational Mathematics,\\
Gaußstrasse 20, 42119 Wuppertal, Germany}

\address[MAIS]{Department of Mathematics, AMNEA Group, Faculty of Sciences and Techniques,\\ Moulay Ismail University of Meknes, Errachidia 52000, Morocco}



\begin{abstract}
This paper introduces a nonstandard finite difference (NSFD) approach to a reaction-diffusion SEIQR 
epidemiological model, which captures the spatiotemporal dynamics of infectious disease transmission. 
Formulated as a system of semilinear parabolic partial differential equations (PDEs), the model extends classical compartmental models by incorporating spatial diffusion to account for population movement and spatial heterogeneity. 
The proposed NSFD discretization is designed to preserve the continuous model's essential qualitative features, such as positivity, boundedness, and stability, which are often compromised by standard finite difference methods. We rigorously analyze the model's well-posedness, construct a structure-preserving NSFD scheme for the PDE system, and study its convergence and local truncation error. 
Numerical simulations validate the theoretical findings and demonstrate the scheme's effectiveness in preserving biologically consistent dynamics. 
\end{abstract}


 \journal{Applied Mathematics and Computation}




\end{frontmatter}



\section{Introduction}\label{S1}
Mathematical modeling has been instrumental in advancing our understanding of infectious disease dynamics and shaping public health interventions. 
Foundational compartmental models, such as the SIR and SEIR frameworks 
\cite{Martcheva}, 
originally introduced by Kermack and McKendrick in the early 20th century \cite{Kermack1927}, have provided essential insights into the spread of epidemics. 
Over the years, these classical models have been extended to incorporate greater epidemiological realism, accounting for factors such as latency periods, waning immunity, 
asymptomatic transmission, and spatial variability  
\cite{EGK19, Zinihi2025FDE}.

Recent public health crises, such as the COVID-19 pandemic, Ebola outbreaks, and the resurgence of diseases like measles and tuberculosis, have renewed interest in the development of more realistic and spatially explicit epidemic models \cite{Zinihi2025FDE}.
In this context, reaction-diffusion systems have emerged as powerful frameworks for capturing both the temporal evolution and spatial spread of infectious diseases within populations \cite{Anita, Capasso1}. 
By introducing diffusion terms into classical compartmental models, these systems enable researchers to model population movement and spatial heterogeneity in infection risk \cite{ZhangC}, which are the key factors in designing effective regional containment strategies. 

In this work we focus on the numerical solution of partial differential equation (PDE) models in epidemiology, which extend ordinary differential equation (ODE) models using spatial reaction-diffusion systems, where each compartment, representing a different species, is allowed to invade a spatial domain $\Omega \subset \mathbb{R}^n$
with a space-dependent density. 
This spatial diffusion mechanism is modelled with the Laplace operator, leading to a system of semilinear parabolic PDEs 
supplied with suitable boundary conditions.
Spatial PDE models have been used to study the transmission of infection, depending on how a particular disease is transmitted between different populations or subpopulations, see e.g.\ 
the SIS reaction-diffusion model in a heterogeneous environment \cite{Allen1}, the modified SIS diffusion models \cite{Huang}, a  reaction-advection-diffusion system with free boundaries \cite{Cheng2021},  or the spatially diffusive SIR epidemic model with the mass action infection mechanism \cite{Kuniya}. 

In a recent study \cite{Zinihi2025S}, the authors proposed a spatiotemporal SEIQR epidemic model that incorporates optimal control strategies. The model accounted for disease transmission in space and time and included three  time- and space-dependent control variables: vaccination for susceptible individuals, social distancing for exposed and infected individuals, and treatment for quarantined individuals. The authors pursued four goals: (i) proving the existence, uniqueness, and positivity of global strong solutions using analytic semigroup theory, (ii) establishing the existence of optimal controls through functional analysis, (iii) deriving first-order necessary conditions via convex perturbations and adjoint equations, and (iv) conducting numerical simulations. 
The numerical results showed that combining pharmaceutical and non-pharmaceutical interventions was more effective in reducing the disease burden and associated control costs.

Conventional finite difference methods often encounter numerical instability and fails to reproduce important qualitative properties of the solutions when applied to differential equations modeling real-world phenomena. 
The nonstandard finite difference (NSFD) approach, pioneered by Mickens \cite{Mickens1, Mickens61a}, offers a more robust alternative to address these limitations. 
The NSFD method is based on carefully formulated design rules that ensure dynamic consistency by preserving key qualitative features of the original continuous system, such as positivity, boundedness and asymptotic behaviour \cite{EhrhardtMickens}.

Numerical solution methods based on NSFD have been widely applied to both ordinary and partial differential equations (PDEs) across various application domains. 
For example, Mickens developed an NSFD scheme for the Burgers equation with a logistic reaction-diffusion term, which preserves the positivity and boundedness of the continuous model \cite{Mickens-Burgers}.
In \cite{Conte}, the authors designed an NSFD scheme
for reaction-diffusion PDEs describing the coexistence of plant species in arid environments. Costa et al.\ \cite{Costa}  proposed an NSFD scheme for solving a 15-component model of the immune response to SARS-CoV-2.


In NSFD schemes, the nonstandard denominator in the discrete derivative reflects the qualitative features of the underlying differential equation.
In epidemiology and ecology, differential equations modeling infectious disease dynamics and predator-prey interactions typically require solutions that remain positive, cf.\ e.g.\ \cite{Costa, Mammeri20, Martcheva}.
As demonstrated in \cite{Maamar} (for a time-fractional model of Zika virus transmission), numerical solutions may become negative when using standard schemes, such as the standard fourth-order Runge-Kutta method.
Additionally, standard schemes can produce numerical fixed-points, which are not actual fixed points of the original ODE model, 
\cite{mickens07}. 

As discussed in \cite{Pasha}, standard finite difference (SFD) methods 
are prone to numerical instability and may fail to preserve essential properties, especially in epidemiological contexts.
Additionally, although various NSFD schemes have been developed for reaction-diffusion systems, many lack temporal accuracy or consistency.
The study addresses these issues by proposing an improved NSFD scheme that guarantees first-order accuracy in time and second-order accuracy in space while ensuring the positivity of the numerical solution.
Several other studies have reinforced the advantages of NSFD approaches in capturing the spatiotemporal dynamics of complex systems. For example, specialized adaptations of Mickens's rules, such as nonlocal approximations of nonlinear terms and nonlinear denominator functions, have been shown to maintain solution properties in cross-diffusion and chemotaxis models \cite{Chapwanya2, deWaal}.
Furthermore, NSFD methods have been successfully applied to fractional reaction-diffusion systems \cite{Taghipour}
and convection-diffusion models \cite{EhrhardtMickens}, demonstrating superior stability and fidelity compared to classical schemes \cite{Chapwanya2}. 
Motivated by these findings, we adopt an NSFD approach in the present work to ensure the positivity and numerical stability of the discretized SEIQR system while accurately capturing the effects of spatial diffusion and nonlinear interactions.

In this paper, we propose a reaction-diffusion SEIQR model governed by a system of parabolic PDEs. 
We assume the population is distributed over a spatially homogeneous domain, and the compartmental densities vary in time and space. 
This captures the spatiotemporal dynamics of disease transmission.
Specifically, the densities of the susceptible, exposed, infectious, quarantined, and recovered individuals are denoted by $S(t,x), E(t,x), I(t,x), Q(t,x),$ and $R(t,x),$ respectively.
This study builds upon our previous work on a spatiotemporal SEIQR model \cite{Zinihi2025S}, which focused on the analysis of optimal control strategies using three intervention variables. 
In contrast, the present work introduces an NSFD scheme to simulate the model's dynamics numerically.
This structure-preserving discretization approach ensures that essential properties of the continuous model, 
such as the positivity and boundedness of the solution, are retained in the numerical setting. 
This provides a more reliable and biologically consistent computational framework.

Our work is structured as follows:  
Section~\ref{S2} introduces the proposed SEIQR epidemic model, which is formulated as a system of PDEs. 
It presents the compartmental transitions and describes the spatiotemporal dynamics of disease spread.
Then, it provides a well-posedness analysis of the model in an $m$-dimensional setting.  
In Section~\ref{S3} we construct an NSFD scheme for the two-dimensional version of the PDE model.
Section~\ref{S4} provides a detailed analysis of the numerical scheme, including convergence to a feasible solution and stability properties. 
Section~\ref{S5} presents numerical simulations that illustrate and validate the theoretical results.  
Finally, Section~\ref{S6} summarizes the main findings and outlines possible directions for future research.

\section{Mathematical Analysis of the Model}\label{S2}
In this section, we propose a five-dimensional epidemiological model to describe the transmission dynamics of an epidemic.  
In this model, individuals in the population transition through five compartments over time: 
Susceptible ($S$), Exposed ($E$), Infected ($I$), Quarantined ($Q$), and Recovered ($R$). 
The transmission coefficients used in the SEIQR model are summarized in Table~\ref{Tab1}.

\begin{table}[htb]
\centering
\setlength{\tabcolsep}{0.8cm}
\caption{Transmission coefficients for the proposed SEIQR reaction-diffusion model.}\label{Tab1}
\adjustbox{max width=\textwidth}{
\begin{tabular}{cc}
\hline 
\textbf{Symbol} & \textbf{Description} \\
\hline \hline 
$\Lambda$ & Recruitment rate (e.g.\ birth or immigration) \\
\hline
\multirow{2}{*}{$\beta_1$} & Transmission rate due to contact between\\ & susceptible and exposed individuals \\
\hline
\multirow{2}{*}{$\beta_2$} & Transmission rate due to contact between\\ & susceptible and infectious individuals \\
\hline
$\mu$ & Natural death rate \\
\hline
$\delta$ & Rate at which exposed individuals become infectious \\
\hline
$\gamma$ & Rate at which infectious individuals are quarantined \\
\hline
$\alpha$ & Recovery rate of quarantined individuals \\
\hline
\multirow{2}{*}{$\rho$} & Rate at which non-infected quarantined individuals\\ & return to the susceptible class \\
\hline
$\lambda_S$, $\lambda_E$, $\lambda_I$, $\lambda_Q$, $\lambda_R$ & Diffusion coefficients for $S$, $E$, $I$, $Q$, and $R$ respectively \\
\hline
\end{tabular}
}
\end{table}

The proposed SEIQR reaction-diffusion model describes the spatiotemporal spread of an infectious disease by modeling the movement and interactions of individuals across five compartments: susceptible ($S$), exposed ($E$), infected ($I$), quarantined ($Q$), and recovered ($R$).\\
\textbf{\textit{Susceptible}} ($S$): Susceptible individuals are at risk of contracting the disease. Their population increases through recruitment $\Lambda$ and re-entry from quarantine at rate $\rho Q$, representing individuals who tested negative or were misclassified. They decrease due to natural mortality ($\mu$) and infection upon contact with exposed ($E$) and infected ($I$) individuals, at rates $\beta_1 S E$ and $\beta_2 S I$, respectively. 
Spatial movement is modeled by the diffusion term $\lambda_S \Delta S$, where $\Delta$ denotes the Laplace operator.\\
\textbf{\textit{Exposed}} ($E$): Exposed individuals have been infected but are not yet infectious. This group increases via contact between susceptibles and exposed individuals, governed by the transmission term $\beta_1 S E$. 
They either progress to the infected class at rate $\delta$ or die naturally. 
Spatial diffusion is represented by $\lambda_E \Delta E$.\\
\textbf{\textit{Infected}} ($I$): Infected individuals are symptomatic and capable of transmitting the disease. Their number increases through contact between susceptibles and infected individuals ($\beta_2 S I$) and through progression from the exposed class ($\delta E$). They may be quarantined at rate $\gamma$ or die naturally at rate $\mu$. Their spatial spread is modeled by $\lambda_I \Delta I$.\\
\textbf{\textit{Quarantined}} ($Q$): This compartment consists of individuals isolated after developing symptoms. They enter from the infected class at rate $\gamma$ and may recover ($\alpha Q$), die naturally ($\mu Q$), or return to the susceptible class ($\rho Q$) if found uninfected. 
Their spatial redistribution is governed by $\lambda_Q \Delta Q$.\\
\textbf{\textit{Recovered}} ($R$): Recovered individuals have acquired immunity after completing quarantine. They accumulate at rate $\alpha Q$ and are subject to natural death at rate $\mu$. Their spatial mobility is described by $\lambda_R \Delta R$.

Throughout the model, the diffusion terms $\lambda_j \Delta X$ for each compartment $X$ reflect spatial spread due to individual movement, while the mortality terms $\mu X$ account for natural deaths unrelated to the disease. 
By incorporating both local disease dynamics and spatial processes, the model captures key mechanisms of epidemic propagation, including localized outbreaks, spatial heterogeneity, and the potential impact of quarantine interventions.\\
Based on the preceding description, Figure~\ref{F1} illustrates the compartmental structure of the SEIQR model, highlighting the key transitions between health states, including infection, quarantine, and recovery.

\begin{figure}[htb]
\centering
\begin{tikzpicture}[node distance=4cm]
\node (S) [rectangle, draw, minimum size=1cm, fill=cyan!30] {S};
\node (I) [rectangle, draw, minimum size=1cm, fill=red!30, right of=S] {I};
\node (Q) [rectangle, draw, minimum size=1cm, fill=blue!30, right of=I] {Q};
\node (E) [rectangle, draw, minimum size=1cm, fill=orange!30, xshift = 4cm, yshift = 2.5cm] {E};
\node (R) [rectangle, draw, minimum size=1cm, fill=green!30, right of=Q] {R};
\draw[->] (-1.5,0) -- ++(S) node[midway,above]{$\Lambda$};
\draw [->] (S.north) -- (0,2.5) -- (E) node[midway,below]{$\beta_1 S E$};
\draw [->] (S) -- (I) node[midway,below]{$\beta_2 S I$};
\draw [->] (4.1,2) -- (4.1,0.5) node[midway,right]{$\delta E$};
\draw [->] (I) -- (Q) node[midway,below]{$\gamma I$};
\draw [->] (Q) -- (R) node[midway,below]{$\alpha Q$};
\draw [->] (7.9,-0.5) -- (7.9,-1.7) -- (0.1,-1.7) node[midway,below]{$\rho Q$} -- (0.1,-0.5);
\draw[->] (-0.1,-0.5) -| (-0.1,-1.5) node[near end,left]{$\mu S$};
\draw[->] (4.5,2.5) -- (5.7,2.5) node[midway,below]{$\mu E$};
\draw[->] (3.9,0.5) -| (3.9,1.5) node[near end,left]{$\mu I$};
\draw[->] (8.1,-0.5) -| (8.1,-1.5) node[near end,right]{$\mu Q$};
\draw[->] (R.south) -| (12,-1.5) node[near end,right]{$\mu R$};
\end{tikzpicture}
\captionof{figure}{Transmission pathways in the proposed SEIQR model.}\label{F1}
\end{figure}
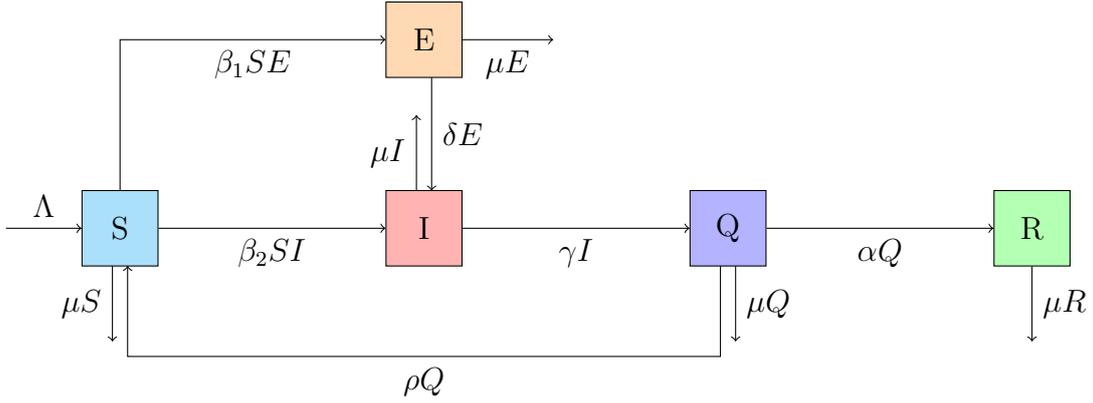

Let $\Omega \subset \mathbb{R}^m$ be a bounded domain with smooth boundary $\partial \Omega$, where $m \in \{1, 2, 3\}$.
\begin{remark}
    Although in general $m \in \mathbb{N}^*$, we restrict $m$ to the values $1 \leq m \leq 3$ to reflect physical reality, as our spatial models are typically embedded in one, two, or three dimensional space. 
    Higher dimensions are not physically meaningful in standard spatial epidemiological or diffusion-based models.
\end{remark}
The spatiotemporal dynamics of the proposed SEIQR model are mathematically described by the following system of reaction-diffusion equations

\begin{align}\label{E2.1}
\left\{ \begin{aligned}
\frac{\partial S(t, x)}{\partial t} 
 &= \lambda_S \Delta S(t, x) + \Lambda + \rho Q(t, x)
 - \beta_1 S(t, x) E(t, x) - \beta_2 S(t, x) I(t, x) - \mu S(t, x), \\
\frac{\partial E(t, x)}{\partial t} &= \lambda_E \Delta E(t, x) + \beta_1 S(t, x) E(t, x) - (\delta + \mu) E(t, x), \\
\frac{\partial I(t, x)}{\partial t} &= \lambda_I \Delta I(t, x) + \beta_2 S(t, x) I(t, x) + \delta E(t, x) - (\gamma + \mu) I(t, x), \\
\frac{\partial Q(t, x)}{\partial t} &= \lambda_Q \Delta Q(t, x) + \gamma I(t, x) - (\alpha + \rho + \mu) Q(t, x), \\
\frac{\partial R(t, x)}{\partial t} &= \lambda_R \Delta R(t, x) +  \alpha Q(t, x) - \mu R(t, x),
\end{aligned}\right.\text{in}\;\mathcal{U}, 
\end{align}
with homogeneous Neumann boundary conditions
\begin{equation}\label{E2.2}
\nabla S\cdot \vec{n} = \nabla E\cdot \vec{n} = \nabla I\cdot \vec{n} = \nabla Q\cdot \vec{n} = \nabla R\cdot \vec{n} = 0, \ \text{ on } \Sigma_T,
\end{equation}
and initial conditions
\begin{equation}\label{E2.3}
S(0,x)=S_0, \ E(0,x)=E_0, \ I(0,x)=I_0, \ Q(0,x)=Q_0, \ R(0,x)=R_0, \ \text{ in } \Omega,
\end{equation}
where $T>0$, $\mathcal{U} = [0, T]\times\Omega$, $\vec{n}$ being the normal to the boundary $\Sigma_T = [0, T]\times\partial\Omega$. 
The imposition of homogeneous Neumann (no-flux) boundary conditions ensures that the SEIQR mode is self-contained, with dynamics driven entirely by internal processes and no movement across the boundary $\partial \Omega$.
In addition, the initial data for all compartments are positive throughout the domain $\Omega$.
Each equation describes the rate of change of a compartment with respect to time and space, capturing the interactions among compartments based on contact rates, transition dynamics, and natural processes such as disease progression and recovery.

The positivity and boundedness of the solutions of a biological-epidemiological model are crucial 
properties. 
Hence, it is necessary to demonstrate that all subpopulations in the system~\eqref{E2.1}--\eqref{E2.3} 
 are bounded and non-negative for all times $t\ge 0$. 
In a previous work \cite{Zinihi2025S} the authors have studied the existence, uniqueness, positivity, and boundedness of the solution to a model similar to
the PDE model~\eqref{E2.1}--\eqref{E2.3}.

\begin{theorem}\label{T1}
Let the initial data for system~\eqref{E2.1}--\eqref{E2.3} be positive. Then, the proposed system admits a unique global solution $\psi = (S, E, I, Q, R)$ that remains positive and bounded on $\mathcal{U}$, with regularity
\begin{equation*}
\psi_i \in W^{1,2}\bigl(0, T; L^2(\Omega)\bigr) \cap L^2\bigl(0, T; H^2(\Omega)\bigr) \cap L^\infty\bigl(0, T; H^1(\Omega)\bigr) \cap L^\infty(\mathcal{U}), \ \forall i \in \{1, 2, 3, 4, 5\}.
\end{equation*}
Moreover, 
\begin{equation*}
\forall i \in \{1, 2, 3, 4, 5\}, \ \Bigl\|\frac{\partial \psi_i}{\partial t}\Bigr\|_{L^2(\mathcal{U})}
+ \bigl\|\psi_i\bigr\|_{L^2(0, T; H^2(\Omega))}
+ \bigl\|\psi_i\bigr\|_{H^1(\Omega)} 
+ \bigl\|\psi_i\bigr\|_{L^{\infty}(\mathcal{U})} < \infty.
\end{equation*}
\end{theorem}

\begin{proof}
The proof is based directly on the methods detailed in \cite[Pages 8--11]{Zinihi2025S}. We outline here the main steps

\textbf{Boundedness:} To establish this property, we introduce a $C_0$-semigroup generated by the Laplace operator. 
Let $\eta$ denote the maximum of the initial data and the second member in each equation. We construct two auxiliary Cauchy problems: one with a shifted upper bound ($+\eta$) and the other with a lower bound ($-\eta$). 
Using semigroup theory and standard a-priori estimates, we deduce that the solution to the reaction-diffusion system remains bounded, i.e. 
\begin{equation*}
S, E, I, Q, R \in L^\infty\bigl(\mathcal{U}\bigr).
\end{equation*}

\textbf{Positivity:} Assuming strictly positive initial data, we use integral estimates and classical comparison principles to show that all solution components remain nonnegative for all time. 
In particular, we prove that the negative part of each variable vanishes, i.e., $S^- = E^- = I^- = Q^- = R^- = 0$, by applying the Cauchy–Schwarz and Gronwall inequalities. 
This ensures the biological feasibility of the model.

\textbf{Existence and Uniqueness:} By exploiting the strong ellipticity of the Laplace operator and the Lipschitz continuity of the second member, we show that the system admits a unique strong solution in the space
\begin{equation*}
S, E, I, Q, R \in W^{1,2}\bigl(0, T; L^2(\Omega)\bigr) \cap L^2\bigl(0, T; H^2(\Omega)\bigr) \cap L^\infty\bigl(0, T; H^1(\Omega)\bigr).
\end{equation*}

For a complete and rigorous treatment of the analysis, the reader is referred to the comprehensive work \cite{Zinihi2025S}, which provides a detailed proof and additional mathematical background.
\end{proof}

The SEIQR reaction-diffusion model~\eqref{E2.1}--\eqref{E2.3} is well-posed in the Hadamard sense:
it admits a unique, global, bounded, and nonnegative solution that depends continuously on the initial data. This provides a solid theoretical foundation for further qualitative and numerical investigations of the model.

\subsection{Basic Reproduction Number}
The basic reproduction number $\mathcal{R}_0$ is a key threshold parameter that determines whether an infectious disease can invade and persist in a susceptible population. 
In this section, we compute 
$\mathcal{R}_0$ for the proposed model using the \textit{next generation matrix} (NGM) approach introduced by \cite{Diekmann1990, vandenDriessche2002}.

\begin{remark}
In this study, the model parameters (including transmission coefficients and recovery rates) are spatially constant. 
Thus, diffusion does not affect the value of $\mathcal{R}_0$. 
In fact, it has been demonstrated that, in this scenario, the basic reproduction number of the reaction-diffusion epidemic model is identical to that of the corresponding ODE model \cite{Yang2020, Yang2023}.
Therefore, the NGM method is sufficient and rigorous for computing 
$\mathcal{R}_0$ in our case. 
In contrast, methods based on integral operators are more suitable when spatial heterogeneity is present in the model parameters.
\end{remark}

The Jacobian matrix of \eqref{E2.1}, evaluated at the point $(S, E, I, Q, R)$, and excluding diffusion effects, is given by
\begin{equation*}
  \mathcal{J} = \begin{pmatrix}
  -\mu - \beta_1 E - \beta_2 I & -\beta_1 S & -\beta_2 S & \rho & 0\\
   \beta_1 E & \beta_1 S - (\mu + \delta) & 0 & 0 & 0\\
  \beta_2 I & \delta & \beta_2 S - (\mu + \gamma) & 0 & 0\\
   0 & 0 & \gamma & -(\mu + \rho + \alpha) & 0\\
    0 & 0 & 0 & \alpha & -\mu
\end{pmatrix}.
\end{equation*}
This matrix is decomposed into the sum of two parts:
\begin{itemize}
\item[$(i)$] The transmission matrix $\mathcal{T}$, which contains all terms representing new infections
\begin{equation*}
  \mathcal{T} = \begin{pmatrix}
  - \beta_1 E - \beta_2 I & -\beta_1 S & -\beta_2 S & 0 & 0\\
   \beta_1 E & \beta_1 S & 0 & 0 & 0\\
   \beta_2 I & 0 & \beta_2 S & 0 & 0\\
   0 & 0 & 0 & 0 & 0\\
    0 & 0 & 0 & 0 & 0
\end{pmatrix}.
\end{equation*}

\item[$(ii)$] The transition matrix $\mathcal{K}$, which captures the rates of transitions between infected classes, as well as outflows due to recovery, death, or movement to other compartments
\begin{equation*}
\mathcal{K} = \begin{pmatrix}
-\mu & 0 & 0 & \rho & 0\\
0 & -(\mu + \delta) & 0 & 0 & 0\\
0 & \delta & -(\mu + \gamma) & 0 & 0\\
0 & 0 & \gamma & -(\mu + \rho + \alpha) & 0\\
0 & 0 & 0 & \alpha & -\mu
\end{pmatrix}.
\end{equation*}
\end{itemize}

In order to apply the NGM method, we must first identify the infected compartments of the model. 
In our case, these are the exposed and infected classes. At the steady-state solution of \eqref{E2.1}, i.e.\ the disease-free equilibrium (DFE) $\bigl(\frac{\Lambda}{\mu}, 0, 0, 0, 0\bigr)$, we extract the relevant $2 \times 2$ submatrices $\widetilde{\mathcal{T}}$ and $\widetilde{\mathcal{K}}$ corresponding to the infected states $E$ and $I$. 
Specifically, the matrices are expressed as
\begin{equation*}
  \widetilde{\mathcal{T}} = 
  \begin{pmatrix}
  \frac{\beta_1 \Lambda}{\mu} & 0 \\
  0 & \frac{\beta_2 \Lambda}{\mu}
   \end{pmatrix} \ \text{ and } \
   \widetilde{\mathcal{K}} = - 
   \begin{pmatrix}
   \mu + \delta & 0 \\
    -\delta & \mu + \gamma
    \end{pmatrix}.
\end{equation*}
Thus, the next generation matrix $\mathcal{N}_\mathcal{GM}$ is computed as
\begin{equation*}
  \mathcal{N}_\mathcal{GM} = -\widetilde{\mathcal{T}}\widetilde{\mathcal{K}}^{-1} =  
\frac{\Lambda}{\mu(\mu + \delta)(\mu + \gamma)} \begin{pmatrix}
\beta_1 (\mu + \gamma) & 0 \\
\beta_2 \delta & \beta_2 (\mu + \delta)
\end{pmatrix}.
\end{equation*}
The basic reproduction number $\mathcal{R}_0$ is defined as the spectral radius of the matrix $\mathcal{N}_\mathcal{GM}$. Because the matrix is triangular, its eigenvalues are simply the diagonal entries. Therefore, we obtain
\begin{equation*}
   \mathcal{R}_0 = \max\Bigl\{\frac{\beta_1\Lambda}{\mu(\mu + \delta)}, \ \frac{\beta_2\Lambda}{\mu(\mu + \gamma)} \Bigr\}.
\end{equation*}
This expression reveals the contributions of exposed and infected individuals to disease transmission. 
The term $\tfrac{\beta_1\Lambda}{\mu(\mu + \delta)}$ reflects the average number of secondary infections produced by an exposed individual during their latency period. 
The term $\tfrac{\beta_2\Lambda}{\mu(\mu + \gamma)}$ corresponds to the same quantity for infected individuals.

\begin{remark}
In general, the basic reproduction number $\mathcal{R}_0$ characterizes a disease's potential to invade and persist in a population.
\begin{enumerate}
\item[$i.$] If $R_0 < 1$, the disease cannot invade the population, and the disease-free equilibrium is locally asymptotically stable.

\item[$ii.$] If $R_0 > 1$, the disease can spread within the population, potentially leading to an endemic equilibrium.

\item[$iii.$] If $\mathcal{R}_0 = 1$, the system is at a threshold, and nonlinear effects may determine whether the disease dies out or persists. Bifurcation analysis is often required to fully understand the dynamics.
\end{enumerate}
\end{remark}


\section{Nonstandard Finite Difference Scheme}\label{S3}
This section outlines the principles of 
NSFD schemes. 
NSFD methods are designed to preserve the key qualitative properties of the original differential equations, such as positivity and conservation laws. 
Consequently, the numerical solutions remain bounded and stable. 
For a detailed account of the discretization strategy, see Mickens's foundational work \cite{Mickens1}.

For simplicity, we consider 1D and 2D settings with a uniform grid, $x_j$ and $y_l$, and a spatial step size, $h$, where $h = \Delta x = \Delta y$. We use the standard notation for the pointwise numerical approximation, e.g., $S_{j,l}^n \approx S(t_n, x_j, y_l)$.

\subsection{NSFD Schemes}\label{subsec:nsfd}
Here, an NSFD scheme is formulated to satisfy the fundamental and essential non-negativity condition. 
Let us recall that for $\Lambda=\mu=0$ the system~\eqref{E2.1}--\eqref{E2.3} reduces to a system of 5 decoupled pure heat equations with known exact solution that can be compared to the NSFD solution.
%
Finally, in Section~\ref{subsec:denom}, we will discuss the so-called denominator function $\phi(k)$. 

We recall that a computational 
scheme for a first-order differential equations system is referred to as a \textit{nonstandard finite difference (NSFD) scheme} if it satisfies at least one of the conditions outlined by Mickens \cite{Mickens1}.


\begin{itemize}

\item Order Consistency:
    The order of the discrete derivatives must match the order of the corresponding derivatives in the original differential equations.

\item Nontrivial Denominator Functions:
    Discrete approximations of derivatives typically involve nontrivial denominator functions. Specifically:
\begin{itemize}
\item First-order derivatives are typically simulated using a generalized forward difference (a modified forward Euler method), given by:
        \begin{equation*}
        \dt{u}\Big|_{t=t_n} \approx \frac{u^{n+1} - u^n}{\phi(k)},
        \end{equation*}
        where $u^n\approx u(t_n)$, on a uniform time grid $t_n=n\,k$, for all $n=0,1,2\dots$, with 
        $k=\Delta t$. 

\item The function $\phi\equiv\phi(k)>0$, known as the \textit{denominator function}, satisfies the consistency condition $\phi(k)=k+\mathcal{O}(k^2)$.
        It is chosen to ensure that the discrete solution exhibits the same asymptotic behavior as the continuous one (s detailed in Section~\ref{subsec:denom}).
        \end{itemize}

\item Nonlocal Approximations of Nonlinear Terms:
    Nonlinear terms are discretized using nonlocal representations, i.e.\ functions involving multiple grid points. Examples include:
    $u^3(t_n) \approx (u^n)^2 \,u^{n+1}$ or $u^2(t_n) \approx u^n\,u^{n+1}$.

\item Preservation of Qualitative Properties:
    Any special conditions satisfied by the continuous model or its solution should also hold for the discrete model and its solution. These include, for example, positivity, convexity (as in financial models), and the preservation of equilibrium points and their local asymptotic stability.
\end{itemize}

\begin{remark}
More generally, derivatives in NSFD schemes are approximated by expressions of the form \cite{Mickens1}:
\begin{equation}\label{eq:discrete}
    \dt{u(t)}\Big|_{t=t_n} \to \frac{u^{n+1}-\psi(k)u^n}{\phi(k)} \, ,
\end{equation} 
where $\psi(k) = 1 + \mathcal{O}(k)$. This generalized time discretization aims to accurately capture the long-term asymptotic behavior of the solution.
\end{remark}

\subsection{NSFD Scheme for the PDE System}
We introduce the following NSFD scheme for solving the PDE model~\eqref{E2.1}--\eqref{E2.3}

\begin{align*}
    \frac{S_{j,l}^{n+1}-S_{j,l}^n}{\phi(k)} 
    &= \lambda_S \Delta_h^2 S_{j,l}^{n} +\Lambda  + \rho Q_{j,l}^n
    - (\beta_1 E_{j,l}^n + \beta_2 I_{j,l}^n +\mu) S_{j,l}^{n+1}\\
    \frac{E_{j,l}^{n+1}-E_{j,l}^n}{ \phi(k) } 
    &= \lambda_E \Delta_h^2 E_{j,l}^n + \beta_1 E_{j,l}^n S_{j,l}^{n+1} - (\delta + \mu) E_{j,l}^{n+1}, \\
    \frac{I_{j,l}^{n+1}-I_{j,l}^n}{ \phi(k) } 
    &= \lambda_I \Delta_h^2 I_{j,l}^n + \beta_2 I_{j,l}^n S_{j,l}^{n+1} + \delta E_{j,l}^{n+1} - (\gamma + \mu) I_{j,l}^{n+1}, 
     \stepcounter{equation}\tag{\theequation}\label{eq:model3}\\
    \frac{Q_{j,l}^{n+1}-Q_{j,l}^n}{\phi(k) } 
    &= \lambda_Q \Delta_h^2 Q_{j,l}^n + \gamma I_{j,l}^{n+1} - (\alpha + \rho + \mu)Q_{j,l}^{n+1}, \\
    \frac{R_{j,l}^{n+1}-R_{j,l}^n}{ \phi(k) } 
    &= \lambda_R \Delta_h^2 R_{j,l}^n +  \alpha Q_{j,l}^{n+1} - \mu R_{j,l}^{n+1},
\end{align*}  
with the denominator function
\begin{equation} \label{eq:phi_nor_c} 
    \phi(k) =\frac{e^{\mu k}-1}{\mu}.
\end{equation}%
In the NSFD scheme \eqref{eq:model3} $\Delta_h^2$ denotes the nonstandard discretization of the Laplacian, proposed by Ben-Charpentier and Kojouharov~\cite{Chen-Charpentier13}, i.e.
\begin{equation}\label{eq:LaplacianNSFD}
    \Delta_h^2 S_{j,l}^{n}
    = \frac{S_{j+1,l}^n - 2S_{j,l}^{n+1} + S_{j-1, l}^n}{h^2} + \frac{S_{j,l+1}^n - 2S_{j,l}^{n+1} + S_{j,l-1}^n}{h^2}.
\end{equation}
While this approach, which is common in NSFD schemes for parabolic PDEs, and preserves the positivity of the solution, this 'skew' discretization introduces a undesirable coupling between space and time errors. By Taylor expansion we can write the truncation error of the nonstandard discretization of the Laplacian \eqref{eq:LaplacianNSFD} as
\begin{equation}\label{eq:LaplacianNSFDc}
    \Delta_h^2 S_{j,l}^{n}
    - \Delta S(t_n,x_j,y_l)
    =-4r\pp{S}{t}\Big|_{(x_j,y_l)}^{t_n}
    - 4kr\ppp{S}{t}\Big|_{(x_j,y_l)}^{t_n}
     +\frac{k}{12r}\Bigl(\frac{\partial^4 S}{\partial {x}^4}\Big|_{(x_j,y_l)}^{t_n}+\frac{\partial^4 S}{\partial {y}^4}\Big|_{(x_j,y_l)}^{t_n}\Bigr)+\gO(k^2),
\end{equation}
with the constant mesh ratio $r=k/h^2$. The leading error term is zeroth-order in $k$, which means that it will not vanish as $k\to0$. So, this discrete operator \eqref{eq:LaplacianNSFD} is not consistent with the continuous Laplacian $\Delta S$. 
This NSFD Laplacian \eqref{eq:LaplacianNSFD} sacrifices consistency in favor of qualitative properties like positivity.
 We will study this topic in more detail in Section~\ref{S4}.

Let us briefly comment on the discretization of the nonlinear terms, which are quadratic in this case. 
For instance, in the first line of \eqref{eq:model3}, the nonlinear contact term $\beta_2 I(t,x)S(t,x)$ from \eqref{E2.1} is discretized as $\beta_2 I_{j,l}^n S_{j,l}^{n+1}$, rather than $I_{j,l}^n S_{j,l}^n$ or $I_{j,l}^{n+1} S_{j,l}^{n+1}$. The guiding principle is that exactly one factor corresponding to the variable with a time derivative ($S$) must be evaluated at the new time level $n+1$. 
This approach ensures that the resulting scheme preserves positivity, as shown in \eqref{eq:model2}.

To maintain an explicit sequential computation, all other variables in a given term are taken from the previous time level unless they have already been updated in earlier equations. 
When applicable, discrete conservation properties 
should be preserved in the discretization.

Although the initial system \eqref{eq:model3} may appear implicit, it can be reformulated into an explicit scheme. Each variable at 
($n+1$) 
can be computed directly using known values from previous time steps, in the order of the equations in the system. In 2D the NSFD scheme \eqref{eq:model3} is reformulated as follows:
\begin{align*}
    S_{j,l}^{n+1} &= \frac{S_{j,l}^n 
    +\lambda_S r(k) (S_{j-1,l}^{n}+S_{j+1,l}^{n}+S_{j,l-1}^{n}+S_{j,l+1}^{n})   
    +\phi(k) (\Lambda +\rho Q_{j,l}^n)}
    {1 +4\lambda_S r(k) + \phi(k) (\beta_1 E_{j,l}^n + \beta_2 I_{j,l}^n +\mu) }, \\
    E_{j,l}^{n+1} &= \frac{E_{j,l}^n
    +\lambda_E r(k) (E_{j-1,l}^{n}+E_{j+1,l}^{n}+E_{j,l-1}^{n}+E_{j,l+1}^{n}) 
    + \phi(k) \beta_1 E_{j,l}^n S_{j,l}^{n+1} } 
    {1 +4\lambda_E r(k) +\phi(k)(\delta + \mu) }, \\
    I_{j,l}^{n+1} &= \frac{I_{j,l}^n
    +\lambda_I r(k) (I_{j-1,l}^{n}+I_{j+1,l}^{n}+I_{j,l-1}^{n}+I_{j,l+1}^{n}) 
    + \phi(k)(\beta_2 I_{j,l}^n S_{j,l}^{n+1} + \delta E_{j,l}^{n+1} )} 
    {1 +4\lambda_I r(k) +\phi(k)(\gamma + \mu) }, 
     \stepcounter{equation}\tag{\theequation}\label{eq:model2}\\
    Q_{j,l}^{n+1} &= \frac{Q_{j,l}^n 
    +\lambda_Q r(k) (Q_{j-1,l}^{n}+Q_{j+1,l}^{n}+Q_{j,l-1}^{n}+Q_{j,l+1}^{n})
    + \phi(k) \gamma I_{j,l}^{n+1}} 
    {1 +4\lambda_Q r(k) + \phi(k)(\alpha + \rho + \mu) },  \\
    R_{j,l}^{n+1} &= \frac{R_{j,l}^n
    +\lambda_R r(k)  (R_{j-1,l}^{n}+R_{j+1,l}^{n}+R_{j,l-1}^{n}+R_{j,l+1}^{n}) 
    + \phi(k)\alpha Q_{j,l}^{n+1}}
    {1+4\lambda_R r(k) +\phi(k) \mu },
\end{align*} 
where we have introduced the \textit{generalized parabolic mesh ratio} $r(k)=\phi(k)/h^2$.

\begin{remark} In one space dimension this NSFD scheme \eqref{eq:model2} reads
\begin{align*}
    S_{j}^{n+1} &= \frac{S_{j}^n
    +\lambda_S r(k) (S_{j-1}^{n}+S_{j+1}^{n})  
    + \phi(k)(\Lambda +\rho Q_{j}^n)}
    {1+2\lambda_S r(k) + \phi(k)(\beta_1 E_{j}^n + \beta_2 I_{j}^n +\mu) }, \\
    E_{j}^{n+1} &= \frac{E_{j}^n
    +\lambda_E r(k) (E_{j-1}^{n}+E_{j+1}^{n}) 
    + \phi(k) \beta_1 E_{j}^n S_{j}^{n+1} } 
    {1+2\lambda_E r(k) + \phi(k) (\delta + \mu) }, \\
    I_{j}^{n+1} &= \frac{I_{j}^n 
    +\lambda_I r(k) (I_{j-1}^{n}+I_{j+1}^{n}) 
    + \phi(k) (\beta_2 I_{j}^n S_{j}^{n+1} + \delta E_{j}^{n+1} )} 
    {1+2 \lambda_I r(k)  + \phi(k) (\gamma + \mu) }, 
     \stepcounter{equation}\tag{\theequation}\label{eq:model21d}\\
    Q_{j}^{n+1} &= \frac{Q_{j}^n
    +\lambda_Q r(k) (Q_{j-1}^{n}+Q_{j+1}^{n}) 
    + \phi(k) \gamma I_{j}^{n+1} } 
    {1+2\lambda_Q r(k) + \phi(k)(\alpha + \rho + \mu) },  \\
    R_{j}^{n+1} &= \frac{R_{j}^n
    +\lambda_R r(k)  (R_{j-1}^{n}+R_{j+1}^{n})
    + \phi(k) \alpha Q_{j}^{n+1} }
    {1+2\lambda_R r(k) +\phi(k)\mu }. 
\end{align*}  
\end{remark}

The computations must be carried out in the specified order. 
In epidemic modeling, it is standard practice to assume that all parameters are non-negative, reflecting their real-world interpretations. 
Under this convention, and based on the explicit form in \eqref{eq:model2}, it is straightforward to verify that the scheme preserves positivity, provided that certain natural conditions on the parameters are satisfied.
The special choice of the  denominator function $\phi(k)$ is discussed in Section~\ref{ss:LTE}.


\section{Analysis of the NSFD Scheme}\label{S4}
In this section we will perform an analysis of the proposed NSFD scheme \eqref{eq:model2}. To do so, 
we start with the local truncation error in Section~\ref{ss:LTE}.
Then we discuss in Section~\ref{subsec:denom} the denominator function and
finally, we study in Section~\ref{ss:nonstan_discreteL} the consistency of the nonstandard discretization of the Laplacian.

\subsection{Local Truncation Error}\label{ss:LTE}
First, we apply a Taylor series expansion in time to the components and use the consistency condition of the denominator function, e.g. for the first component
\begin{equation}\label{eq:LTE1}
\begin{split}
    S_{j,l}^{n+1}=S_{j,l}^{n} + k\pp{S}{t}\Big|_{(x_j,y_l)}^{t_n} +\gO(k^2)
                 &=S_{j,l}^{n} + \phi(k)\pp{S}{t}\Big|_{(x_j,y_l)}^{t_n} +\gO(k^2)\\
                 &=S_{j,l}^{n} + r(k)h^2\pp{S}{t}\Big|_{(x_j,y_l)}^{t_n} +\gO(k^2). 
\end{split}
\end{equation}
Inserting in the scheme \eqref{eq:model3} yields
\begin{equation}\label{eq:LTE2}
\begin{split}
    S_{j,l}^n  + \phi(k)\pp{S}{t}\Big|_{(x_j,y_l)}^{t_n}
    &= S_{j,l}^n + \lambda_S r(k)\biggl(S_{j+1,l}^n + S_{j-1, l}^n +S_{j,l+1}^n + S_{j,l-1}^n - 4\Bigl( S_{j,l}^{n} + k\pp{S}{t}\Big|_{(x_j,y_l)}^{t_n}\Bigr) \biggr)\\ 
     &\quad  +\phi(k) \biggl(\Lambda  + \rho Q_{j,l}^n
    - (\beta_1 E_{j,l}^n + \beta_2 I_{j,l}^n +\mu) \Bigl( S_{j,l}^{n} + k\pp{S}{t}\Big|_{(x_j,y_l)}^{t_n}\Bigr)\biggr)+\gO(k^2).
\end{split}
\end{equation}
Next, we use the fact, that on the right hand side \eqref{eq:LTE2} we have a standard second order in space semi-discretization of the first component of the PDE, i.e. 
\begin{equation}\label{eq:LTE3}
\begin{split}
   \phi(k)\pp{S}{t}\Big|_{(x_j,y_l)}^{t_n}
    &= \phi(k)\pp{S}{t}\Big|_{(x_j,y_l)}^{t_n}
    - 4\lambda_S r(k) k\pp{S}{t}\Big|_{(x_j,y_l)}^{t_n} \\ 
     &\qquad  -\phi(k) k\bigl(\beta_1 E_{j,l}^n + \beta_2 I_{j,l}^n +\mu\bigr) \Bigl(\pp{S}{t}\Big|_{(x_j,y_l)}^{t_n}\Bigr)+\gO(k^2).
\end{split}
\end{equation}
Comparing coefficients of time derivative and dividing by $\phi(k)$ yields 
\begin{equation}\label{eq:LTE4}
 1 = 1 - 4\lambda_S \frac{k}{h^2}  - k\bigl(\beta_1 E_{j,l}^n + \beta_2 I_{j,l}^n +\mu\bigr),
\end{equation}
and similarly for the other components, with different diffusion constants $\lambda$. 
Hence, since the parabolic mesh ratio $k/h^2$ is assumed to be constant, we conclude that the NSFD scheme \eqref{eq:model3} cannot be consistent of first order in time, due to the discretization \eqref{eq:LaplacianNSFD}.

\subsection{The Denominator Function}\label{subsec:denom}

Pasha, Nawaz and Arif \cite{Pasha} proposed a remedy for this above mentioned order reduction. They considered the condition
\eqref{eq:LTE4} in the form
\begin{equation}\label{eq:LTE4b}
 k = \phi_S(k) \bigl(1 - 4\lambda_S \frac{k}{h^2}  - k(\beta_1 E_{j,l}^n + \beta_2 I_{j,l}^n +\mu)\bigr).
\end{equation}
This resulted in a denominator function for the first component $S$:
\begin{equation}\label{eq:LTE4c}
 \phi_S(k) = k \bigl(1 - 4\lambda_S \frac{k}{h^2}  - k(\beta_1 E_{j,l}^n + \beta_2 I_{j,l}^n +\mu)\bigr)^{-1},
\end{equation}
and accordingly for the other components.
However, this choice \eqref{eq:LTE4c} depends on the component (and also on the location, $x_j$ and $y_l$, and even time instance $t_n$), which prevents discrete conservation properties for the total population when the individual equations for the components are summed up. 
With the parabolic mesh ratio $r=k/h^2$ we rewrite \eqref{eq:LTE4c} to
\begin{equation}\label{eq:LTE4cc}
 \phi_S(k) = \frac{k}{1 - 4r\lambda_S  - k(\beta_1 E_{j,l}^n + \beta_2 I_{j,l}^n +\mu)}
 =\frac{k}{a-kb_{j,l}^n}
\end{equation}
with $a=1 - 4r\lambda_S$ and $b=b_{j,l}^n=\beta_1 E_{j,l}^n + \beta_2 I_{j,l}^n +\mu$.
Now, we apply a Neumann series expansion for small $k$, assuming $k|b_{j,l}^n| < |a|$:
\begin{equation}\label{eq:LTE4D}
 \frac{1}{a-kb_{j,l}^n}=\frac{1}{a}\frac{1}{1-\frac{kb}{a}}
 =\frac{1}{a}\sum_{q=0}^\infty \Bigl(\frac{kb}{a}\Bigr)^q
 =\frac{1}{a}+\frac{kb}{a^2}+\frac{(kb)^2}{a^3}+\dots.
\end{equation}
This gives the denominator function
\begin{equation}\label{eq:LTE4DD}
\phi_S(k)
 =\frac{k}{1-4r\lambda_S}+\frac{\beta_1 E_{j,l}^n + \beta_2 I_{j,l}^n +\mu}{(1-4r\lambda_S)^2}k^2+\frac{(\beta_1 E_{j,l}^n + \beta_2 I_{j,l}^n +\mu)^2}{(1-4r\lambda_S)^3}k^3+\dots,
\end{equation}
which shows that the consistency condition $\phi(k)=k+\mathcal{O}(k^2)$ is satisfied,
if the time step constraint $k< |a|/|b_{j,l}^n|$ is fulfilled.

Analogue calculations can be done for the four other components in the system \eqref{eq:model2} and thus we can formulate similar to \cite[Theorem~1]{Pasha}
\begin{theorem}
    The NSFD scheme \eqref{eq:model2} constructed for the system~\eqref{E2.1}--\eqref{E2.3} 
    using
    \begin{equation}\label{eq:LTE5}
          \phi_\ell(k) = k \bigl(1 - 4\lambda_\ell \frac{k}{h^2} -k\mu 
          -g_\ell(k)\bigr)^{-1},
    \end{equation}
    where $r=k/h^2=\text{const.}$ and 
    \begin{equation}
    g_\ell(k)=\begin{cases}
         k(\beta_1 E_{j,l}^n + \beta_2 I_{j,l}^n) & \text{if}\;\ell=S,\\  
         k\delta & \text{if}\;\ell=E,\\ 
         k\gamma & \text{if}\;\ell=I,\\ 
         k(\alpha + \rho) & \text{if}\;\ell=Q,\\ 
         0& \text{if}\;\ell=R,
    \end{cases}
   \end{equation}
   assuming the time step restrictions $|k\mu+g_\ell(k)| < |1 - 4r\lambda_\ell|$, $\ell=1,2,\dots,5$, are fulfilled,
   has the consistency order 1 in time and 2 in space.
\end{theorem}
However, this is a circular argument, since this particular choice of denominator function results in the standard explicit finite difference method, i.e., no unconditional stability and no unconditional positivity of the solution.

\subsection{Consistency of the discretized Laplacian}\label{ss:nonstan_discreteL}
In order to overcome the described consistency issues of the discretized Laplacian
\eqref{eq:LaplacianNSFD} let us include a perturbation $\delta(k)$ as follows
\begin{equation}\label{eq:LaplacianNSFDdelta}
    \Delta_h^2 S_{j,l}^{n}
    = \frac{S_{j+1,l}^n - 2\bigl(1+\delta(k)\bigr)\Psi(k)S_{j,l}^{n+1} + S_{j-1, l}^n}{h^2} + \frac{S_{j,l+1}^n - 2\bigl(1+\delta(k)\bigr)\Psi(k)S_{j,l}^{n+1} + S_{j,l-1}^n}{h^2}.
\end{equation}
%
Assuming a constant mesh ratio $r = k/h^2$, the truncation error becomes
\begin{align*}
 \Delta_h^2 S_{j,l}^{n}
    - \Delta S(t_n,x_j,y_l) &=  - 4r\bigl(1 + \delta(k)\bigr) \pp{S}{t}\Big|_{(x_j,y_l)}^{t_n} - \frac{4r \delta(k)}{k} S(t_n,x_j,y_l)
    +\frac{k}{12r}\Delta^2 S\Big|_{(x_j,y_l)}^{t_n}
    +\gO(k^2).
\end{align*}


In order to avoid the divergent second $1/k$-term, we must require $\delta(k)=\gO(k)$,
say $\delta(k)=\nu k+\gO(k^2)$. We plug this into the parts of the truncation error that 
are not vanishing as $k\to0$:
\begin{equation}
 - 4r\bigl(1 + \nu k\bigr) \pp{S}{t}\Big|_{(x_j,y_l)}^{t_n} - 4 \nu r S(t_n,x_j,y_l).
\end{equation}
This still has an $\gO(1)$-component, unless we cancel both leading-order terms:
\begin{equation}
 0 = 4r\pp{S}{t}\Big|_{(x_j,y_l)}^{t_n} + 4 \nu r S(t_n,x_j,y_l), \quad\Rightarrow\quad \nu =-\frac{S_t}{S}\Big|_{(x_j,y_l)}^{t_n}.
\end{equation}
This again gives a non-universal expression for $\nu$ -- it depends on the solution
$S$, and its time derivative $S_t$, which can be computed numerically as 
\begin{equation}
   \nu=\nu(S_{j,l}^n) = -\frac{S_{j,l}^n-S_{j,l}^{n-1}}{kS_{j,l}^n},\;n=1,2,\dots,\quad \nu(S_{j,l}^0)=0,
\end{equation}
and accordingly for the other four components.
Hence, we finally end up with the following modification of the NSFD scheme \eqref{eq:model2}, which reads
\begin{align*}
    S_{j,l}^{n+1} &= \frac{S_{j,l}^n 
    +\lambda_S r(k) (S_{j-1,l}^{n}+S_{j+1,l}^{n}+S_{j,l-1}^{n}+S_{j,l+1}^{n})   
    +\phi(k) (\Lambda +\rho Q_{j,l}^n)}
    {1 +4\lambda_S r(k) (1+\nu(S_{j,l}^n) k) + \phi(k) (\beta_1 E_{j,l}^n + \beta_2 I_{j,l}^n +\mu) }, \\
    E_{j,l}^{n+1} &= \frac{E_{j,l}^n
    +\lambda_E r(k) (E_{j-1,l}^{n}+E_{j+1,l}^{n}+E_{j,l-1}^{n}+E_{j,l+1}^{n}) 
    + \phi(k) \beta_1 E_{j,l}^n S_{j,l}^{n+1} } 
    {1 +4\lambda_E r(k)(1+\nu(E_{j,l}^n) k) +\phi(k)(\delta + \mu) }, \\
    I_{j,l}^{n+1} &= \frac{I_{j,l}^n
    +\lambda_I r(k) (I_{j-1,l}^{n}+I_{j+1,l}^{n}+I_{j,l-1}^{n}+I_{j,l+1}^{n}) 
    + \phi(k)(\beta_2 I_{j,l}^n S_{j,l}^{n+1} + \delta E_{j,l}^{n+1} )} 
    {1 +4\lambda_I r(k)(1+\nu(I_{j,l}^n) k) +\phi(k)(\gamma + \mu) }, 
     \stepcounter{equation}\tag{\theequation}\label{eq:model2mod}\\
    Q_{j,l}^{n+1} &= \frac{Q_{j,l}^n 
    +\lambda_Q r(k) (Q_{j-1,l}^{n}+Q_{j+1,l}^{n}+Q_{j,l-1}^{n}+Q_{j,l+1}^{n})
    + \phi(k) \gamma I_{j,l}^{n+1}} 
    {1 +4\lambda_Q r(k)(1+\nu(Q_{j,l}^n) k) + \phi(k)(\alpha + \rho + \mu) },  \\
    R_{j,l}^{n+1} &= \frac{R_{j,l}^n
    +\lambda_R r(k)  (R_{j-1,l}^{n}+R_{j+1,l}^{n}+R_{j,l-1}^{n}+R_{j,l+1}^{n}) 
    + \phi(k)\alpha Q_{j,l}^{n+1}}
    {1+4\lambda_R r(k) (1+\nu(R_{j,l}^n) k)+\phi(k) \mu }.
\end{align*} 
We will use the approach \eqref{eq:model2mod} for our numerical experiments in Section~\ref{S5}.

\begin{remark} In one space dimension this NSFD scheme \eqref{eq:model2mod} reads
\begin{align*}
    S_{j}^{n+1} &= \frac{S_{j}^n 
    +\lambda_S r(k) (S_{j-1}^{n}+S_{j+1}^{n})   
    +\phi(k) (\Lambda +\rho Q_{j}^n)}
    {1 +2\lambda_S r(k) (1+\nu(S_{j}^n) k) + \phi(k) (\beta_1 E_{j}^n + \beta_2 I_{j}^n +\mu) }, \\
    E_{j}^{n+1} &= \frac{E_{j}^n
    +\lambda_E r(k) (E_{j-1}^{n}+E_{j+1}^{n}) 
    + \phi(k) \beta_1 E_{j}^n S_{j}^{n+1} } 
    {1 +2\lambda_E r(k)(1+\nu(E_{j}^n) k) +\phi(k)(\delta + \mu) }, \\
    I_{j}^{n+1} &= \frac{I_{j}^n
    +\lambda_I r(k) (I_{j-1}^{n}+I_{j+1}^{n}) 
    + \phi(k)(\beta_2 I_{j}^n S_{j}^{n+1} + \delta E_{j}^{n+1} )} 
    {1 +2\lambda_I r(k)(1+\nu(I_{j}^n) k) +\phi(k)(\gamma + \mu) }, 
     \stepcounter{equation}\tag{\theequation}\label{eq:model2mod1d}\\
    Q_{j}^{n+1} &= \frac{Q_{j}^n 
    +\lambda_Q r(k) (Q_{j-1}^{n}+Q_{j+1}^{n})
    + \phi(k) \gamma I_{j}^{n+1}} 
    {1 +2\lambda_Q r(k)(1+\nu(Q_{j}^n) k) + \phi(k)(\alpha + \rho + \mu) },  \\
    R_{j}^{n+1} &= \frac{R_{j}^n
    +\lambda_R r(k)  (R_{j-1}^{n}+R_{j+1}^{n}) 
    + \phi(k)\alpha Q_{j}^{n+1}}
    {1+2\lambda_R r(k) (1+\nu(R_{j}^n) k)+\phi(k) \mu }. 
\end{align*}  
\end{remark}

\section{Numerical Results}\label{S5}

In this section, we present numerical simulations that validate and illustrate the theoretical results established in previous sections. These simulations demonstrate the effectiveness of the proposed numerical schemes in preserving biologically meaningful dynamics.
We perform the simulations for one and two-dimensional spatial domains to analyze the spatiotemporal behavior of the SEIQR reaction-diffusion model.
We aim to verify the preservation of essential properties, such as positivity and boundedness, and to compare the performance of the SFD and NSFD methods.
Accordingly, this section is divided into two parts. The first part focuses on one-dimensional simulations, and the second part addresses the two-dimensional case.
For both settings, we use a denominator function $\phi(k) = k = 10^{-4}$ and a spatial step size $h = 0.01$. Hence, the considered parabolic mesh ratio is $r=1$.

\subsection{One-Dimensional Simulations}
In the one-dimensional simulations, we used specific parameter values and initial conditions inspired by the following previous works:  \cite{Pei2009, Zinihi2025DB, Verma2021, ZhouM, Isik2025, Zinihi2025S, Ma2022, Hwang2022}. These are detailed below
\begin{center}
\begin{minipage}[t]{.23\textwidth}
\begin{itemize}
\item $\Lambda = 1$,

\item $\beta_1 = 0.001$,

\item $\beta_2 = 0.003$,
\end{itemize}
\end{minipage}
\hfill
\begin{minipage}[t]{.23\textwidth}
\begin{itemize}
\item $\mu = 0.02$,

\item $\delta = 0.01$,

\item $\gamma = 0.04$,
\end{itemize}
\end{minipage}
\hfill
\begin{minipage}[t]{.45\textwidth}
\begin{itemize}
\item $\alpha = 0.04$,

\item $\rho = 0.02$,

\item $\lambda_S = \lambda_E = \lambda_I = \lambda_Q = \lambda_R = 0.1$.
\end{itemize}
\end{minipage}
\end{center}
The spatial domain $\Omega = \{x \in \mathbb{R} \ | \ 0\leq x \leq 1\}$ is discretized into $N_x + 1$ equidistant points defined by $x_j = jh$ for $j = 0, 1, \dots, N_x$, where $N_x = \frac{1}{h}$ is the uniform spatial step size. 
For the NSFD method, we use the scheme \eqref{eq:model2mod1d}. 
In the case of the SFD method, the diffusion terms are approximated using the standard second-order finite difference formula. 
For instance, for the variable $S$, we use
\begin{equation*}
    \Delta_h S_j^n \approx \frac{S_{j+1}^n - 2S_j^n + S_{j-1}^n}{h^2},
\end{equation*}
and the time derivative is approximated by the forward difference quotient
\begin{equation*}
    \frac{\partial S_j^n}{\partial t} \approx \frac{S_j^{n+1} - S_j^n}{k}.
\end{equation*}

Figure~\ref{F2} illustrates the evolution of the epidemiological classes $S(t,x)$, $E(t,x)$, $I(t,x)$, $Q(t,x)$, and $R(t,x)$ over both space and time using the SFD method.
\begin{itemize}
\item The 3D plots (left column) show the spatiotemporal dynamics for each compartment.

\item The 2D plots (right column) represent the temporal evolution of the spatially averaged populations $S(t, \cdot)$, $E(t, \cdot)$, $I(t, \cdot)$, $Q(t, \cdot)$, and $R(t, \cdot)$.
\end{itemize}

\begin{figure}[H]
\centering
\includegraphics[width=1\textwidth]{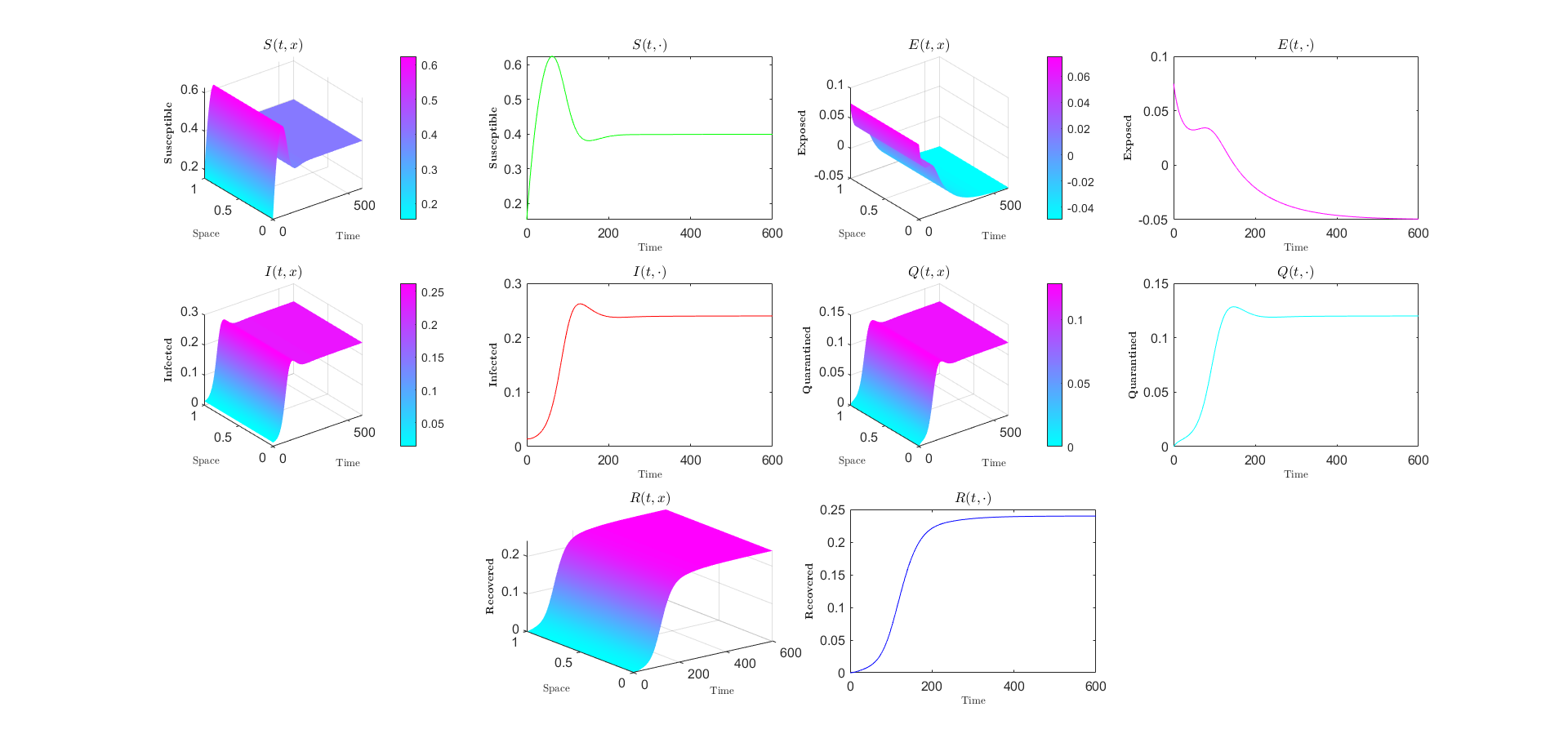}
\caption{Spatiotemporal dynamics (1D) of the SEIQR model~\eqref{E2.1}--\eqref{E2.3} using the SFD scheme.}\label{F2}
\end{figure}
Although the general qualitative behavior of the epidemic is captured, an important issue arises in the exposed class, $E(t,x)$.
As can be seen in the corresponding 3D and 2D plots, the solution becomes negative at later times. This is biologically implausible, as population densities cannot be negative.
This drawback is a known limitation of classical explicit schemes such as SFD, which may violate positivity and boundedness unless stringent stability conditions are met.

Figure~\ref{F3} shows the same model solved using the NSFD method, which is designed to preserve key dynamical properties such as positivity and boundedness.
\begin{itemize}
\item The 3D and 2D plots for each compartment demonstrate the epidemic's smooth and realistic evolution.

\item Most notably, the exposed class $E(t,x)$ remains strictly nonnegative throughout the entire simulation period.
\end{itemize}

\begin{figure}[htb]
\centering
\includegraphics[width=1\textwidth]{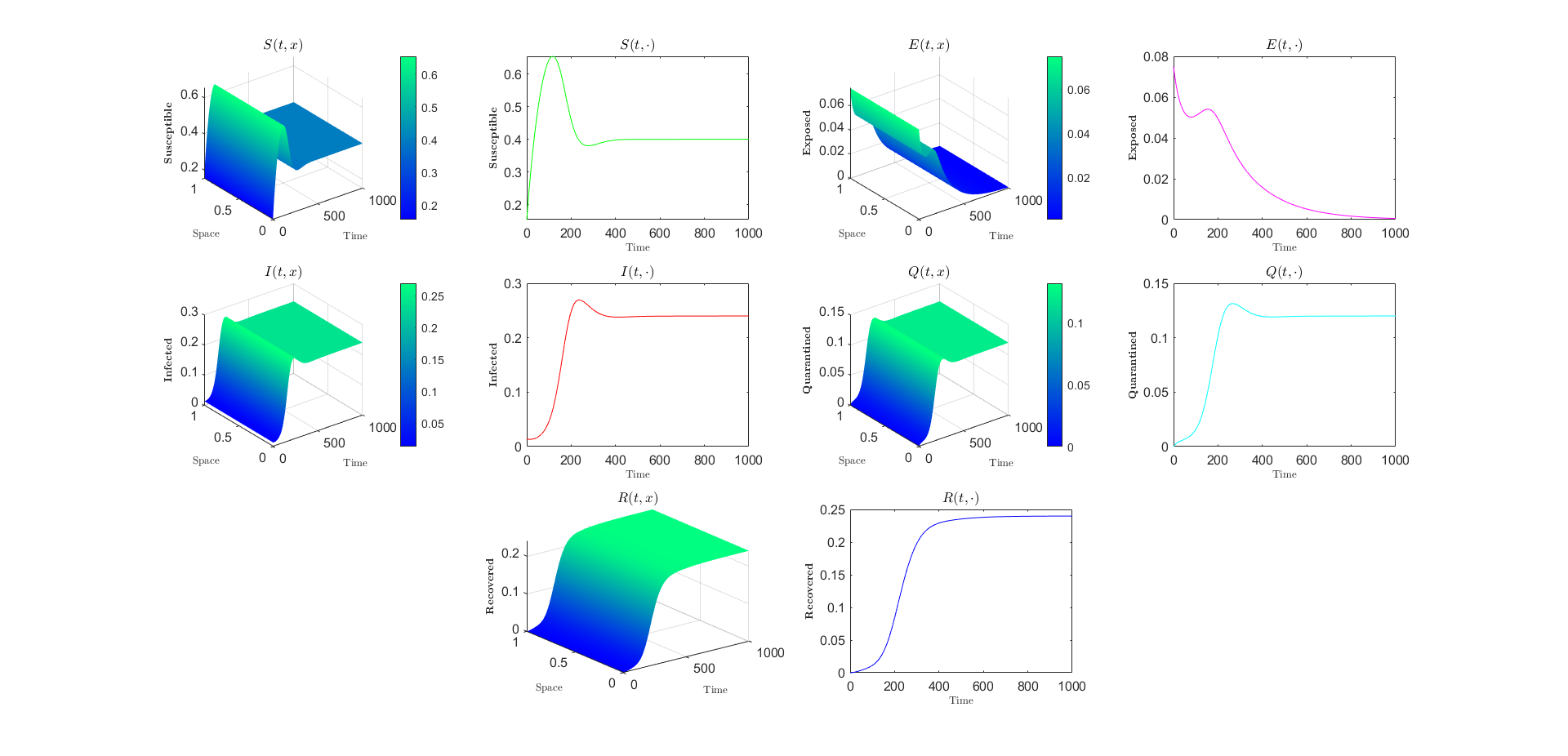}
\caption{Spatiotemporal dynamics (1D) of the SEIQR model~\eqref{E2.1}--\eqref{E2.3} using the NSFD scheme.}\label{F3}
\end{figure}
The NSFD scheme maintains the biological integrity of the model while avoiding the numerical artifacts, such as negative populations, observed in the SFD case in Figure~\ref{F2}.

\subsection{Two-Dimensional Simulations}
In this case, we consider the spatial domain $\Omega = [0,1] \times [0,1]$. 
At the initial time $t = 1$, the population is assumed to be uniformly distributed, with 100\,\% susceptible individuals per cell throughout $\Omega$, except in the center of the domain. 
In this central subdomain, the local population is initialized as 70\,\% susceptible, 20\,\% exposed, and 10\,\% infected.
The parameter values used in the simulations are selected based on the literature \cite{Zinihi2025S, RuizHerrera2021, Gerberry2008, Bolzoni2021, Fu2024, Li2022, Guo2023}, and are given below
\begin{center}
\begin{minipage}[t]{.23\textwidth}
\begin{itemize}
\item $\Lambda = 1$,

\item $\beta_1 = 0.06$,

\item $\beta_2 = 0.07$,
\end{itemize}
\end{minipage}
\hfill
\begin{minipage}[t]{.23\textwidth}
\begin{itemize}
\item $\mu = 0.06$,

\item $\delta = 0.05$,

\item $\gamma = 0.04$,
\end{itemize}
\end{minipage}
\hfill
\begin{minipage}[t]{.45\textwidth}
\begin{itemize}
\item $\alpha = 0.05$,

\item $\rho = 0.03$,

\item $\lambda_S = \lambda_E = \lambda_I = \lambda_Q = \lambda_R = 0.01$.
\end{itemize}
\end{minipage}
\end{center}
For the numerical simulations, we apply the NSFD scheme described in~\eqref{eq:model2mod}. In the case of the SFD method, the diffusion terms are approximated using the formula~\eqref{eq:LaplacianNSFD}. 
For example, the time derivative of the variable $S$ is approximated by the forward-in-time difference
\begin{equation*}
    \frac{\partial S_{j,l}^n}{\partial t} \approx \frac{S_{j,l}^{n+1} - S_{j,l}^n}{k}.
\end{equation*}

Figure~\ref{F4} presents the two-dimensional spatiotemporal evolution of the SEIQR model compartments using the classical SFD method.
The dynamics are shown as heat maps; each subplot represents the population density distribution across the two-dimensional spatial domain at a fixed time.
Initially ($t=1$), the spatial spread originates from a concentrated source. 
By $t=20$ to $t=100$, the infection progresses and spreads spatially.
However, a notable numerical flaw appears in the exposed class $E(t, x, y)$: negative values emerge at later times, especially from $t = 80$ to $t = 100$.
This is evident visually from the darker (non-physical) regions and is confirmed by the color scale, which extends below zero.
This behavior contradicts biological reality because population densities cannot be negative. 
It also indicates a limitation in the classical SFD scheme, especially for reaction-diffusion models, where preserving positivity is essential.
\begin{figure}[htb]
\centering
\includegraphics[width=1\textwidth]{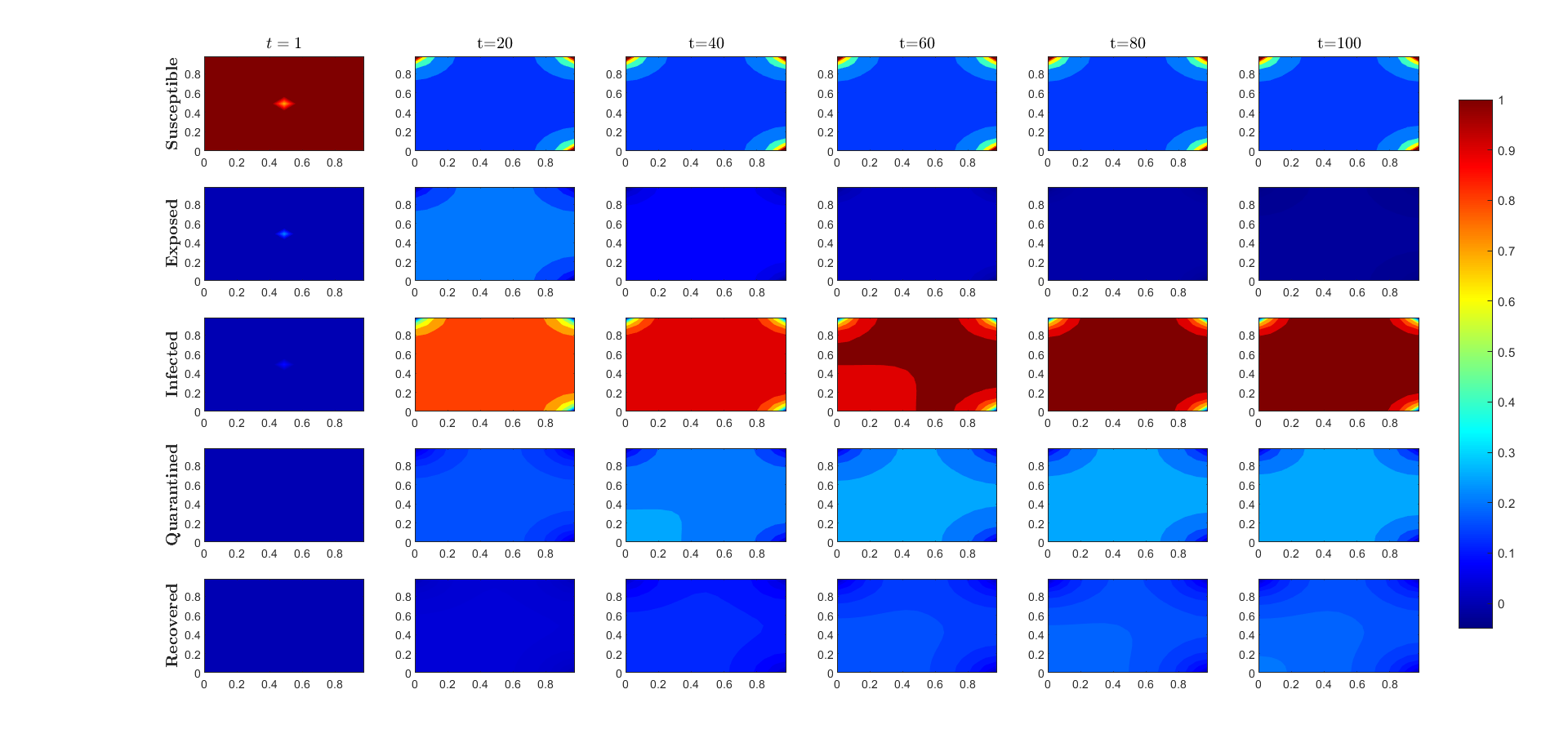}
\caption{Spatiotemporal evolution (2D) of the SEIQR model~\eqref{E2.1}--\eqref{E2.3} using the SFD scheme.}\label{F4}
\end{figure}

In Figure~\ref{F5}, the two-dimensional spatiotemporal dynamics of the reaction-diffusion SEIQR model~\eqref{E2.1}--\eqref{E2.3} are simulated using the NSFD method. 
In this case, the spatial distributions remain smooth and biologically meaningful across time for all compartments.
Notably, the exposed class maintains non-negative values throughout the entire simulation period, thus avoiding the numerical breakdown observed in the SFD case.
All variables exhibit consistent evolution, capturing both the infection dynamics and spatial propagation without generating nonphysical artifacts.
\begin{figure}[htb]
\centering
\includegraphics[width=1\textwidth]{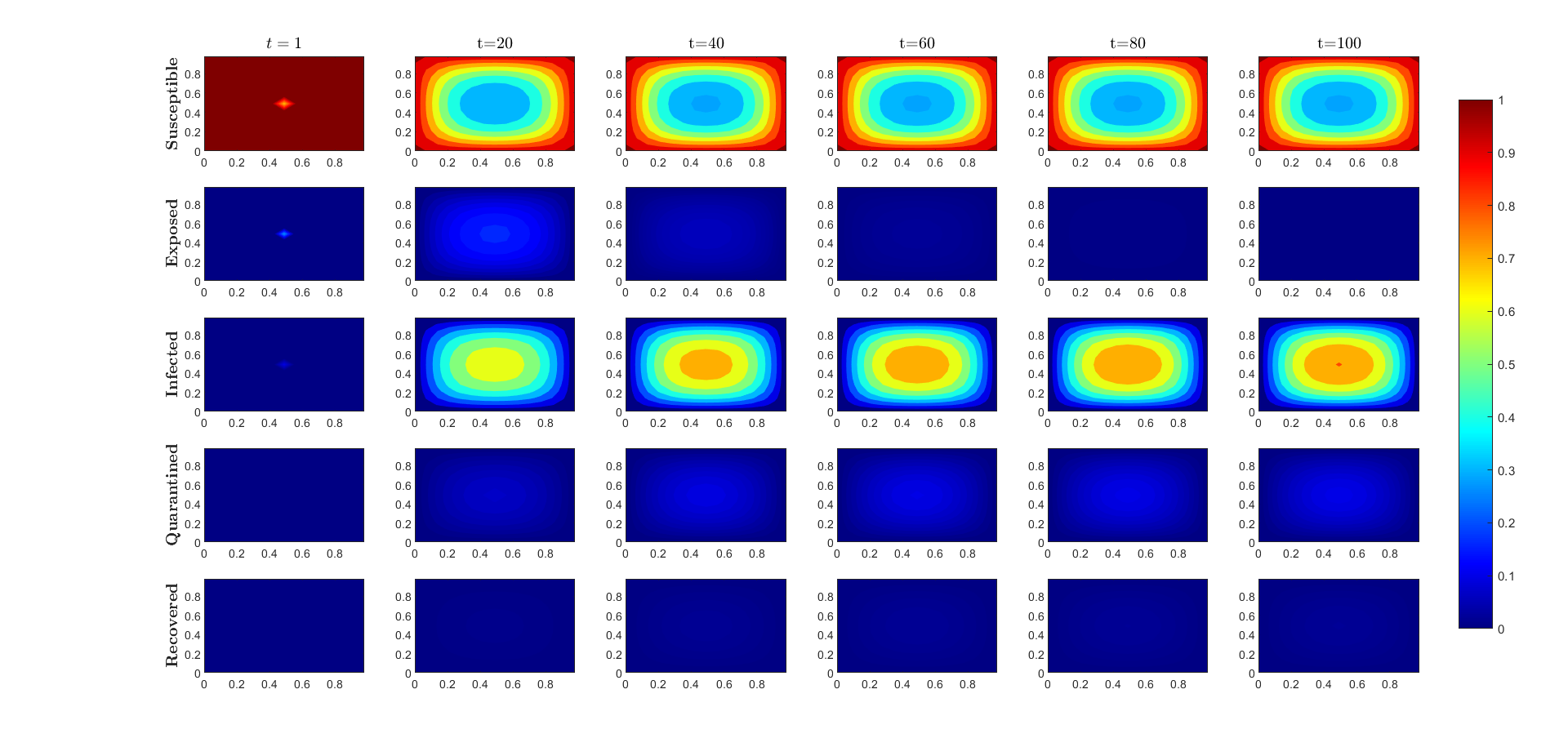}
\caption{Spatiotemporal evolution (2D) of the SEIQR model~\eqref{E2.1}--\eqref{E2.3} using the NSFD scheme.}\label{F5}
\end{figure}

\section{Conclusion and Future Work}\label{S6}
In this work, we proposed and analyzed an SEIQR reaction-diffusion epidemiological model to investigate the spatiotemporal dynamics of infectious disease spread.
Formulated as a system of semilinear parabolic PDEs, the model extends classical compartmental frameworks by incorporating spatial diffusion to reflect population mobility and spatial heterogeneity.
To overcome the limitations of SFD schemes, we developed an NSFD method that rigorously preserves the continuous model's key properties, such as positivity, boundedness, and stability.

We proved the well-posedness of the model in an $m$-dimensional domain and constructed an NSFD scheme for one- and two-dimensional spatial domains.
A detailed numerical analysis confirmed the convergence and stability of the proposed method.
Further confirmation of the theoretical results came from numerical simulations, which demonstrated that, unlike SFD schemes, the NSFD approach avoids nonphysical artifacts, such as negative population values. This makes the NSFD approach more suitable for realistic epidemiological modeling.

Looking ahead, there are several meaningful ways to further enrich the current modeling framework and enhance its relevance to real-world scenarios. One promising approach is to integrate fractional diffusion or nonlocal operators, which can more accurately capture the long-range spatial interactions and memory effects inherent in complex biological systems.
Expanding the model to heterogeneous or time-dependent spatial domains would provide a more realistic representation of environmental variability and demographic structure. Additionally, incorporating optimal control strategies, such as vaccination, treatment, social distancing, and public awareness campaigns, could provide valuable insights into designing efficient, targeted intervention policies.

Additionally, future developments may include stochastic formulations, which account for uncertainties and random fluctuations in transmission dynamics. These formulations are particularly useful in the early stages of outbreaks or in small populations.
Using fractional-order dynamics embeds memory effects directly into the system, aligning it more closely with observed epidemiological behaviors.
Finally, finance-oriented extensions grounded in actuarial science and healthcare economics \cite{Feng2022} could bridge the gap between epidemiological modeling and decision-making processes under resource constraints. These extensions could support epidemic insurance schemes, healthcare capacity planning, and the optimal allocation of limited resources.

\section*{Declarations}

\subsection*{Data availability} 
All information analyzed or generated, which would support the results of this work are available in this article.
No data was used for the research described in the article.

\subsection*{Conflict of interest} 
The authors declare that there are no problems or conflicts 
of interest between them that may affect the study in this paper.



\begin{thebibliography}{99}














\bibitem{Allen1}
L. J. S. Allen, B. M. Bolker, Y. Lou, A. L. Nevai, 
Asymptotic profiles of the steady states for an SIS epidemic reaction-diffusion model,  
Discrete and Continuous Dynamical Systems - B  21 (2008) 1-20.











\bibitem{Anita}
S. Anita, V. Capasso, 
Reaction-Diffusion Systems in Epidemiology, 
An. Stiint. Univ. Al. I. Cuza Iasi. Mat. (N.S.) Tomul LXVI, 2020, f. 2.






\bibitem{Bolzoni2021}
L. Bolzoni, R. Della Marca, M. Groppi,
On the optimal control of SIR model with Erlang-distributed infectious period: isolation strategies,
Journal of Mathematical Biology 83 (2021).



















\bibitem{Capasso1}
V. Capasso, 
Reaction-Diffusion Models for the Spread of a Class of Infectious Diseases. 
In: H. Neunzert (eds), Proceedings of the Second European Symposium on Mathematics in Industry. 
European Consortium for Mathematics in Industry 3, 1988, Springer, Dordrecht. 





\bibitem{Chapwanya2}
M. Chapwanya, J. M.-S. Lubuma, R. E. Mickens, 
Nonstandard finite difference schemes for Michaelis- Menten type reaction‐diffusion equations,
Numerical Methods for Partial Differential Equations 29(1) (2013) 337-360.




\bibitem{Chen-Charpentier13}
B. M. Chen-Charpentier, H. V. Kojouharov, 
An unconditionally positivity preserving scheme for advection–diffusion reaction equations,
Mathematical and Computer Modelling 57(9-10) (2013) 2177-2185.


\bibitem{Cheng2021}
C. Cheng, Z. Zheng, 
Dynamics and spreading speed of a reaction-diffusion system with advection modeling West Nile virus, 
Journal of Mathematical Analysis and Applications 493 (2021) 124507.


\bibitem{Conte}
D. Conte, G. Pagano, B. Paternoster, 
Nonstandard finite differences numerical methods for a vegetation reaction-diffusion model,
Journal of Computational and Applied Mathematics 419 (2023) 114790.




\bibitem{Costa}
G. M. R. Costa, M. Lobosco, M. Ehrhardt, R. F. Reis, 
Mathematical Analysis and a Nonstandard Scheme for a Model of the Immune Response against COVID-19, 
Mathematical and Computational Modeling of Phenomena Arising in Population Biology and Nonlinear Oscillations, 
AMS Contemporary Mathematics, 2023.


%


\bibitem{deWaal}
G. N. de Waal, A. R. Appadu, C. J. Pretorius,
Some standard and nonstandard finite difference schemes for a reaction–diffusion–chemotaxis model,
Open Physics 21(1) (2023) 20220231.

\bibitem{Diekmann1990}
O. Diekmann, J. A. P. Heesterbeek, J. A. J. Metz,
On the definition and the computation of the basic reproduction ratio $R_0$ in models for infectious diseases in heterogeneous populations,
Journal of Mathematical Biology 28 (1990).




\bibitem{EhrhardtMickens}
M. Ehrhardt, R. E. Mickens, 
A nonstandard finite difference scheme for convection-diffusion equations having constant coefficients, 
Applied Mathematics and Computation 219 (2013) 6591-6604.


\bibitem{EGK19}
M. Ehrhardt, J. Gašper, S. Kilianová, 
SIR-based Mathematical Modeling of Infectious Diseases with Vaccination and Waning Immunity, 
Journal of Computational Science 37 (2019) 101027. 

\bibitem{Feng2022}
R. Feng, J. Garrido, L. Jin, S. H. Loke, L. Zhang, 
Epidemic Compartmental Models and Their Insurance Applications,
in: M. d. C. Boado-Penas, J. Eisenberg, S. Şahin (eds.),
Pandemics: Insurance and Social Protection, Springer Actuarial. Springer, Cham, 2022.

\bibitem{Fu2024}
X. Fu, J. Wang,
Dynamic behaviors and non-instantaneous impulsive vaccination of an SAIQR model on complex networks,
Applied Mathematics and Computation 465 (2024) 128425.

\bibitem{Gerberry2008}
D. J. Gerberry, F. A. Milner,
An SEIQR model for childhood diseases,
Journal of Mathematical Biology 59 (2008) 535–561.

\bibitem{Guo2023}
Y. Guo, Z. Liu, L. Wang, R. Tan,
Modeling the role of information and optimal control on an SEIR epidemic model with infectivity in latent period,
Mathematical Methods in the Applied Sciences 47 (2023) 1044–1064.






\bibitem{Huang}
W. Huang, M. Han, K. Liu, 
Dynamics of an SIS reaction-diffusion epidemic model for disease transmission, 
Mathematical Biosciences and Engineering 7 (2010) 51-66.

\bibitem{Hwang2022}
Y. Hwang, H. D. Kwon, J. Lee,
Optimal control problem of various epidemic models with uncertainty based on deep reinforcement learning,
Numerical Methods for Partial Differential Equations 38 (2022) 2142–2162.

\bibitem{Isik2025}
O. R. Isik, N. Tuncer, M. Martcheva,
A mathematical model for the role of vaccination and treatment in measles transmission in Turkey,
Journal of Computational and Applied Mathematics 457 (2025) 116308.






\bibitem{Kermack1927}
W. O. Kermack, A. G. McKendrick, 
A contribution to the mathematical theory of epidemics, 
Proceedings of the Royal Society of London - Series A 115 (1927) 700-721.








\bibitem{Kuniya}
T. Kuniya, J. Wang, 
Lyapunov functions and global stability for a spatially diffusive SIR epidemic model, 
Applicable Analysis 96 (2017) 1935-1960.

\bibitem{Li2022}
Z. Li, T. Zhang,
Analysis of a COVID-19 epidemic model with seasonality,
Bulletin of Mathematical Biology 84 (2022).










\bibitem{Ma2022}
Y. Ma, Y. Cui, M. Wang,
Global stability and control strategies of a SIQRS epidemic model with time delay,
Mathematical Methods in the Applied Sciences 45 (2022) 8269–8293.


\bibitem{Maamar}
M. H. Maamar, M. Ehrhardt, L. Tabharit,
A nonstandard finite difference scheme for a time-fractional model of Zika virus transmission,
Mathematical Biosciences and Engineering 21(1) (2024) 924-962. 



\bibitem{Mammeri20}
Y. Mammeri, 
A reaction-diffusion system to better comprehend the unlockdown: Application of SEIR-type model with diffusion to the spatial spread of COVID-19 in France,
Computational and Mathematical Biophysics 8(1) (2020) 102-113.




\bibitem{Martcheva}
M. Martcheva, 
An Introduction to Mathematical Epidemiology, 
Springer New York, NY, 2015.







\bibitem{Mickens1}
R. E. Mickens,
Nonstandard Finite Difference Models of Differential Equations, 
World Scientific, Singapore, 1994.


\bibitem{Mickens61a}
R. E. Mickens, 
Nonstandard finite difference schemes for reaction-diffusion equation, 
Numerical Methods for Partial Differential Equation 15 (1999) 201-214.




\bibitem{Mickens-Burgers}
R. E. Mickens,
A nonstandard finite difference scheme for the diffusionless Burgers
equation with logistic reaction,
Mathematics and Computers in Simulation 62(1–2) (2003), 117-124.








\bibitem{mickens07} 
R. E. Mickens,
Calculation of denominator functions for nonstandard finite difference schemes for differential equations satisfying a positivity condition,  
Numerical Methods for Partial Differential Equation 23(3)(2007), 672-691.





\bibitem{Pasha}
S. A. Pasha, Y. Nawaz, M. S. Arif,
On the nonstandard finite difference method for reaction-diffusion models, 
Chaos, Solitons \& Fractals 166 (2023) 112929.





\bibitem{Pei2009}
Y. Pei, S. Liu, S. Gao, S. Li, C. Li, 
A delayed SEIQR epidemic model with pulse vaccination and the quarantine measure, 
Computers \& Mathematics with Applications  58 (2009) 135–145.









\bibitem{RuizHerrera2021}
A. Ruiz-Herrera, P. J. Torres,
The role of movement patterns in epidemic models on complex networks,
Bulletin of Mathematical Biology 83(10) (2021) Article No. 98.







\bibitem{Taghipour}
M. Taghipour, H. Aminikhah,
An Efficient Non-standard Finite Difference Scheme for Solving Distributed Order Time Fractional Reaction–Diffusion Equation,
International Journal of Applied and Computational Mathematics 8(2) (2022) 56.

\bibitem{vandenDriessche2002}
P. van den Driessche, J. Watmough,
Reproduction numbers and sub-threshold endemic equilibria for compartmental models of disease transmission,
Mathematical Biosciences 180 (2002) 29-48.

\bibitem{Verma2021}
T. Verma, A. K. Gupta,
Network synchronization,  stability and rhythmic processes in a diffusive mean-field coupled SEIR model,
Communications in Nonlinear Science and Numerical Simulation 102 (2021) 105927.












\bibitem{Yang2020}
C. Yang, J. Wang,
Basic reproduction numbers for a class of reaction-diffusion epidemic models,
Bulletin of Mathematical Biology 82 (2020).

\bibitem{Yang2023}
C. Yang, J. Wang,
Computation of the basic reproduction numbers for reaction-diffusion epidemic models,
Mathematical Biosciences and Engineering 20 (2023) 15201–15218.







\bibitem{ZhangC}
C. Zhang, J. Gao, H. Sun, J. Wang,
Dynamics of a reaction-diffusion SVIR model in a spatial heterogeneous environment, 
Physica A 533 (2019) 122049.


\bibitem{ZhouM}
M. Zhou, H. Xiang, Z. Li, 
Optimal control strategies for a reaction-diffusion epidemic system, 
Nonlinear Analysis: Real World Applications 46 (2019) 446-464.


\bibitem{Zhu2020}
C.-C. Zhu, J. Zhu, 
Spread trend of COVID-19 epidemic outbreak in China: using exponential attractor method in a spatial heterogeneous SEIQR model. 
Mathematical Biosciences and Engineering 17(4) (2020), 3062-3087.


\bibitem{Zinihi2025S}
A. Zinihi, M. Ehrhardt, M. R. Sidi Ammi,
Spatiotemporal SEIQR epidemic modeling with optimal control for vaccination,  treatment, and social measures, 
arXiv preprint, arXiv:2507.09328, 2025.

\bibitem{Zinihi2025DB}
A. Zinihi, M. R. Sidi Ammi, M. Ehrhardt,
Dynamical behavior of a stochastic epidemiological model: stationary distribution and extinction of a SIRS model with stochastic perturbations, 
SeMA Journal, (2025).

\bibitem{Zinihi2025FDE}
A. Zinihi, M. R. Sidi Ammi, D. F. M. Torres,
Fractional differential equations of a reaction-diffusion SIR model involving the Caputo-fractional time-derivative and a nonlinear diffusion operator,
Evolution Equations and Control Theory 14 (2025) 944–967.

\end{thebibliography}

\bibliographystyle{amsalpha}

\end{document}